\documentclass[a4paper, 11pt]{article}
\usepackage{fullpage}

\usepackage{t1enc}
\usepackage[utf8]{inputenc}

\usepackage{latexsym,amssymb,amsmath,amsthm,graphicx,enumerate, mathtools}
\usepackage{hyperref,cite,color, comment,tikz}

\newcommand{\is}{\textsc{independent set}}
\newcommand{\maxsmp}{\textsc{max-smpo}}
\newcommand{\kernel}{\textsc{max-kernel-po}}
\newcommand{\maxsmti}{\textsc{max-smti}}

\newcommand{\Rset}{\ensuremath{\mathbb R}}

\newcommand{\cI}{\ensuremath{\mathcal I}}

\theoremstyle{plain}
\newtheorem{thm}{Theorem}
\newtheorem{theorem}[thm]{Theorem}
\newtheorem{proposition}[thm]{Proposition}
\newtheorem{lemma}[thm]{Lemma}
\newtheorem{claim}[thm]{Claim}

\theoremstyle{definition}

\newtheorem{remark}[thm]{Remark}
\newtheorem{example}[thm]{Example}

\title{Maximum Stable Matching with Matroids and Partial Orders
\thanks{Some of the results in this paper have appeared in the 6th SIAM Symposium on Simplicity of Algorithms (SOSA 2023) \cite{csaji2022approximation}}}

\author{
Gergely Csáji\thanks{Department of Operations Research, Eötvös Loránd University, Budapest, Hungary,  and Mechanism Design Research Group, Institute of Economics and Regional Studies, Budapest, Hungary. Email: \texttt{csaji.gergely@krtk.hu}, \texttt{csajigergely@student.elte.hu}} 
\and
Tamás Király\thanks{ELKH-ELTE Egerv\'ary Research Group, Department of Operations Research, E\"otv\"os Loránd University, Budapest, Hungary. Email: \texttt{tamas.kiraly@ttk.elte.hu}}
\and
Yu Yokoi\thanks{Department of Mathematical and Computing Science,
School of Computing, Tokyo Institute of Technology, Tokyo, Japan. 
Email: \texttt{yokoi@c.titech.ac.jp}}}

\begin{document}

\maketitle

\begin{abstract}
  The Stable Marriage problem (SM), solved by the famous deferred acceptance algorithm of Gale and Shapley (GS), has many natural generalizations. If we allow ties in preferences, then the problem of finding a maximum stable matching becomes NP-hard, and the best known approximation ratio is $1.5$ (McDermid 2009, Paluch 2011, Z. Király 2012), achievable by running GS on a cleverly constructed modified instance. Another elegant generalization of SM is the matroid kernel problem introduced by Fleiner (2001), which is solvable in polynomial time using an abstract matroidal version of GS. Our main result is a simple $1.5$-approximation algorithm for the matroid kernel problem when preferences are given as interval orders --- a broad subclass of partial orders that covers many applications beyond preferences with ties. In addition, for the bipartite matching case, we show that the output of our algorithm also $1.5$-approximates the LP-optimum of the relaxation of the corresponding Integer Program, which shows that the integrality gap is at most $1.5$ for the interval order case. %\maxsmti\ is at most $1.5$. 
  To contrast this with hardness results, we show that if arbitrary partial orders are allowed in the preferences, then even in the bipartite matching case, the problem becomes hard to approximate within a factor better than $2$ assuming the Unique Games Conjecture, and the integrality gap becomes $2$.  
\end{abstract}

\section{Introduction}
%\todo[inline]{Somewherer in ths introduction, we can mention Gergely's observation that strongly stable matching is hard to find even for semiorders.}
%\todo[inline]{Yu: I'd like to formulate our problem with strict partial orders (i.e., irreflexive, antisymmetric and transitive relation. Like Tamas's popular assignment paper). So I prefer the notation $M_i=(S, \cI_i, \succ_i)$ for a {\em partially ordered matroid} rather than $M_i=(S, \cI_i, \succsim_i)$.}
%\todo[inline]{Tamás: We should use the same font for Max-SMTI and Max-SMP.}
The deferred acceptance algorithm of Gale and Shapley \cite{GS62}  is a quintessential example of a simple combinatorial algorithm that has wide-ranging applications, in such diverse areas as healtchcare labor markets, kidney exchange planning, project allocations, and school choice mechanisms. The original stable marriage problem solved by the Gale--Shapley algorithm has been generalized in many directions, and the mathematical research in the area is still thriving, 60 years after the original paper. 

The aim of this paper is to bring together two directions in which the problem has been extended. One is the design of approximation algorithms for finding a maximum stable matching when ties are allowed in the preference lists. The other is the generalization of the stable marriage problem to matroid intersection, in particular, the matroid kernel problem introduced and solved by Fleiner \cite{fleiner2001matroid} using an abstract version of the Gale--Shapley algorithm.

We show that the best known approximation ratio of $1.5$ for the maximum stable marriage problem with ties \cite{Mcdermid09} can also be achieved for the matroid kernel problem with ties. This extension is made possible by a new theorem on basis exchanges for two disjoint bases of a matroid, which may be of independent interest. 

In addition to extending the approximation result to matroids, we also go beyond preferences with ties, and achieve the same approximation ratio for preferences given by interval orders. Interval orders are partial orders that can be obtained as left-to-right precedence relations of intervals on the real line. Equivalently, they can be characterized by the property that the disjoint union of two comparable pairs contains a third comparable pair. Several stability notions that have appeared in the literature can be modeled as interval orders (e.g., $\alpha$-stability \cite{anshelevich2013anarchy,pini2013stability}, near-stability \cite{chen2021matchings}). We show that a simple variant of the Gale--Shapley algorithm applied to carefully constructed modified instances of the problem gives 
a $1.5$-approximation for the maximum matroid kernel problem with interval orders.

We complement this result with a new hardness result for the maximum stable marriage problem with arbitrary partial orders. We show that, assuming the Unique Games Conjecture \cite{khot2002power}, it is NP-hard to approximate the problem within a factor of $2-\delta$ for any $\delta>0$.

%problems with cardinal preferences and various near-stability requirements. In all of these cases, our algorithms are simple variants of the Gale--Shapley algorithm applied to carefully constructed modified instances of the problem. 

%The key observation that enables the extension of the proof to matroid kernels is an exchange property for ordered matroids that may be of independent interest.

\subsection{Basic definitions}\label{sec:def}

\subsubsection*{Partial orders, weak orders, and stable marriage with ties.} A \emph{(strict) partial order} on a ground set $S$, denoted by $\succ$, is an irreflexive, transitive, and asymmetric binary relation%
\footnote{In literature, the term partial order is also used to refer to a \emph{reflexive}, transitive, and antisymmetric binary relation, which we call a \emph{non-strict partial order}. Strict partial orders and non-strict ones are equivalent in the sense that there is a one-to-one correspondence between them; a strict partial order is converted to a non-strict partial order by adding binary relations $x\succ x$ for all $x\in S$, and vise versa.}
 on $S$. In this paper, we always assume that $|S|$ is finite. If several partial orders are given, we use indices to distinguish them.

Two elements $x$ and $y$ are \emph{incomparable} if $x \not \succ y$ and $x \not \prec y$. A partial order is a \emph{weak order} if being incomparable is a transitive relation. In this case, we use $x\sim y$ to indicate that $x$ and $y$  are incomparable,  or in other words, tied. In a weak order, the ground set $S$ is partitioned into equivalence classes of tied elements.
%on a ground set $S$ is an ordering of $S$ that may contain ties; in other words, it is a partial order where being incomparable is a transitive relation. We use $\succsim$ to denote a weak order; if several weak orders are given, we use indices to distinguish them. The notation $x \sim y$ means that $x$ is tied with~$y$.

In the \emph{stable marriage problem with ties and incomplete lists} (\textsc{smti})%
\footnote{While the \emph{stable marriage problem} was orinally defined for simple bipartite graphs, we use the term to mean a generalized version where graphs can be multigraphs, i.e., parallel edges are allowed.}%
, we are given a bipartite graph $G=(U,W;E)$, and a weak order $\succ_v$ on $\delta_G(v)$ for every vertex $v \in U\cup W$, where $\delta_G(v)$ denotes the set of edges incident to $v$ in $G$. Given a matching $N$ in $G$, an edge $e \in E\setminus N$ with endpoints $u, w$ is a \emph{blocking edge} for $N$ if the following two conditions hold:
\begin{itemize}
    \item $\delta_G(u)\cap N=\emptyset$ or $e \succ_u N(u)$,
    \item $\delta_G(w)\cap N=\emptyset$ or $e \succ_w N(w)$,
\end{itemize}
 where $N(u)$ denotes the edge of $N$ incident to $u$ if it exists. The matching $N$ is \emph{stable} if no edge blocks it. The \maxsmti\ problem is to find a stable matching of maximum size. The problem is NP-hard \cite{IMMM99}, and the best known polynomial-time approximation algorithm has an approximation ratio of $1.5$ \cite{Mcdermid09}. It is also known that no approximation ratio better than $4/3$ is achievable assuming the Unique Games Conjecture (UGC) \cite{yanagisawa2007approximation}. In this paper, we will also consider the \emph{maximum stable marriage problem with partial orders} (\maxsmp), where the preferences are not required to be weak orders. To our knowledge, no stronger hardness of approximation results have previously been known for \maxsmp\ than those for \maxsmti.
%Kiraly13, kiraly2012linear,Paluch14

If no ties are allowed at all in the preference orders (i.e., the preferences are strict), then we obtain the standard stable marriage problem, where all stable matchings have the same size by the so-called ``rural hospitals theorem'' \cite{Roth84, GS85}, and the Gale--Shapley algorithm finds one efficiently. In fact, we can obtain a stable matching of a \maxsmti\ or \maxsmp\ instance by modifying each agent's preferences to an arbitrary total order that is consistent with the agent's partial order, and running the Gale--Shapley algorithm with the resulting total orders. However, the size of the stable matching obtained depends on the choice of total orders, and it can be as small as half the optimum.
%smaller than $2/3$ of the optimum.

\subsubsection*{Matroid kernels.} A natural way to generalize the stable marriage problem and the \maxsmti\ problem is to allow agents to have multiple partners, but have some restrictions on the possible sets of partners. There are several models in this vein: the hospitals-residents problem, the college admissions problem with common quotas, classified stable matchings etc. (see Section \ref{sec:related} for more details). 

From a theoretical point of view, a particularly elegant generalization is the matroid kernel problem defined by Fleiner \cite{fleiner2001matroid}. 
We use the notation and terminology of Schrijver \cite[Chapter~39]{schrijver2003combinatorial} for matroids, unless otherwise stated. See Section~\ref{sec:improving} for the definitions of some basic notions of matroids (such as independence, bases, fundamental circuits).

A \emph{partially ordered matroid} is a triple $(S,\cI,\succ)$ where $S$ is the set of elements, $\cI\subseteq 2^S$ is the family of independent sets of the matroid, and $\succ$ is a partial order on $S$. Let $M_1=(S,\cI_1,\succ_1 )$ and $M_2=(S,\cI_2,\succ_2 )$ be partially ordered matroids on the same ground set $S$.
A common independent set $X \in \cI_1 \cap \cI_2$ is an \emph{$(M_1,M_2)$-kernel} if for every $y \in S \setminus X$ there exists $i \in \{1,2\}$ such that $X+y \notin \cI_i$ and $y \not\succ_i x$ for every $x \in X$ for which $X-x+y \in \cI_i$. If there is an element $y \in S \setminus X$ for which this does not hold, then we say that $y$ \emph{blocks} $X$.
The \kernel\ problem is to find an $(M_1,M_2)$-kernel of maximum size.

If $M_1,M_2$ are partition matroids and $\succ_1,\succ_2$ are weak orders, then the \kernel\ problem is equivalent to \maxsmti. Indeed, we can construct a bipartite graph by considering the partition classes of the two matroids as vertices and the elements of $S$ as edges, whose two endpoints are the vertices corresponding to the two partition classes containing that element. When each vertex has a weak order that is the restriction of $\succ_1$ or $\succ_2$ to its incident edges, then there is a one-to-one correspondence between the stable matchings of this bipartite graph and the $(M_1,M_2)$-kernels.

To see how the matroid kernel problem is considerably more general than the stable marriage problem, observe that in the above construction, we have a 1-uniform matroid for each partition class. We can replace these with arbitrary matroids, and still have a special case of the matroid kernel problem. This corresponds to a generalization of the many-to-many stable matching, where each agent has a matroid on the set of incident edges, and the chosen edge set is required to be independent in each of these matroids. Furthermore, we can include upper bounds on the number of edges incident to subsets of agents on both sides, provided that for each side these upper bounds are given by a polymatroid function.

Fleiner \cite{fleiner2001matroid, fleiner2003fixed} considered the matroid kernel problem without incomparability, i.e., when $\succ_1$ and $\succ_2$ are total orders. He showed that matroid kernels always exist and have
the same size (in fact, they have the same span in both matroids). He also gave a matroidal version of the Gale--Shapley algorithm that finds a matroid kernel efficiently.
 In case of weak orders, the kernels may have different sizes, and it is NP-hard to find a largest one, since this problem is a generalization of \maxsmti.

\subsubsection*{Interval orders.} Interval orders form a useful and well-studied subclass of partial orders. To give an intuitive definition, let the ground set $S$ be a family of intervals on the real line: $S=\{[a_1,b_1],\dots, [a_k,b_k]\}$. Then, we can define a partial order $\succ$ on $S$ by $[a_i,b_i] \succ [a_j,b_j] \Leftrightarrow a_i>b_j$ (i.e., the $i$th interval is completely to the right of the $j$th interval). An \emph{interval order} is a partial order that can be obtained in this way. An alternative definition, which we will later use, is that an interval order is a partial order that has no induced sub-poset of size 4 isomorphic to two disjoint chains of size 2 (with no other relations between the 4 elements). In other words, 
%for any two disjoint comparable pairs, there is a third comparable pair among the 4 elements. 
if there are distinct elements $x,y,z, w\in S$ with $x\succ y$ and $z\succ w$, then we have $x\succ w$ or $z\succ y$.
Due to this property, interval orders are sometimes called $(2+2)$-free posets.

Weak orders are obviously interval orders. A larger notable subclass is that of {\em semiorders}, where a numerical utility value is assigned to each element of $S$, and two elements are considered incomparable if their utilities differ by at most a fixed threshold $\delta$. This corresponds to an interval order where all intervals have the same length. In the context of matching problems with preferences, semiorders have been used to model the reluctance of agents to switch partners if the improvement is too small; see e.g.\ the notion of locally nearly stable matchings in \cite{chen2021matchings}. Interval orders provide a more refined way to model this kind of reluctance by allowing the incomparability intervals to have different sizes.

\subsection{Our contribution}

The key tool for generalizing the $1.5$-approximation to matroid kernels is a result on the existence of a perfect matching of certain types of exchangeable pairs in a matroid. This is presented as Theorem \ref{thm:matching} in Section \ref{sec:improving}. This result may be of independent interest, as it is a previously unknown elegant property of basis exchange graphs.
%Due to its simplicity and apparent usefulness, it is conceivable that this result has already been discovered in other contexts. However, we did not find it in the literature, so it is certainly not as well known as it should be.

Using Theorem \ref{thm:matching}, we show in Section \ref{sec:kernel} that there is a $1.5$-approximation algorithm for \kernel\ with interval orders. That is, we show the following theorem.
\begin{theorem}\label{thm:main-approx}
Given two partially ordered matroids $M_1=(S,\cI_1,\succ_1 )$ and $M_2=(S,\cI_2,\succ_2 )$ with $\succ_1$ and $\succ_2$ being interval orders, one can find an $(M_1, M_2)$-kernel $A$ such that $|A|$ is at least $\frac{2}{3}$ of the size of a maximum $(M_1, M_2)$-kernel.
\end{theorem}
We remark that Theorem~\ref{thm:main-approx} implies a $1.5$-approximation algorithm for finding a maximum solution in the many-to-many stable matching problem on $G=(U,W;E)$ where each vertex $v$ has a matroid constraint and an interval order on $\delta_G(v)$. To reduce the problem to the setting in Theorem~\ref{thm:main-approx}, we construct an interval order $\succ_U$ on the ground set $E$ consistent to the orders of all vertices in $U$ as follows. (The construction of $\succ_W$ is similar). Assign indices to elements in $U$ as $u_1, u_2,\dots, u_{|U|}$ arbitrarily, and then define $\succ_U$ by taking the union of binary relations in $\{\succ_u\}_{u\in U}$ and adding the relations $e\succ_u f$ for all $e\in \delta_G(u_i)$ and $f\in \delta_G(u_j)$ with $i<j$.
We see that $\succ_U$ is an interval order, and its restriction to $\delta_G(u)$ is $\succ_u$ for any $u\in U$. Set $M_1=(E, \cI_U,\succ_U)$ and $M_2=(E, \cI_W, \succ_W)$, where $\cI_U$ and $\cI_W$ are respectively the direct sums of the matroids of vertices in $U$ and $V$. Then $(M_1, M_2)$-kernels correspond to stable matchings on $G$.
%(In matroid terminology, the requirement for each $\succ_i$ inTheorem~\ref{thm:main-approx} can be slightly relaxed the restriction of $\succ_i$ to each connected component of the matroid $(E, \cI_i)$ is an interval order.)

In the case of the stable marriage setting, in Section~\ref{sec:inapprox}, we show an even stronger statement than Theorem~\ref{thm:main-approx}: the output of the algorithm also $1.5$-approximates the LP optimum of the corresponding natural Integer Program. This complements the previously known lower bound of $1.5-o(1)$ of Iwama et al. \cite{iwama201425} on the integrality gap of \maxsmti.

Our algorithm consists of three steps: (1) constructing an instance of the matroid kernel problem with total orders on the extended ground set obtained by replacing each element by three parallel elements, (2) running Fleiner's algorithm on the new instance, and (3) projecting the solution to the original ground set. %Furthermore, given cardinal preferences and a threshold $\Delta$, we show that the same algorithm can be used to find a $1.5$-approximation for the maximum size $\Delta$-min-stable common independent set. 
The running time of the algorithm is quadratic in $|S|$ (provided that the independence oracles of matroids are available), and linear in $|S|$ in case of partition matroids, which corresponds to the many-to-many stable matching problem with parallel edges allowed. Since interval orders can model $\alpha$-stability \cite{anshelevich2013anarchy,pini2013stability} and near-stability \cite{chen2021matchings}, our algorithm implies $1.5$-approximation for these stability notions.

The approximation ratio attained by our algorithm is best possible under UGC. As explained at the end of Section~\ref{sec:apxproof}, this fact easily follows from the result of Askalidis et al. \cite{askalidis2013socially}, who proved that \textsc{max-smiss},  the {\em maximum stable marriage problem with incomplete lists under social stability}, is inapproximable within $\frac{3}{2}-\delta$ under UGC. As \textsc{max-smiss} can be seen as a special case of the maximum stable marriage problem with interval orders, the following proposition follows.
\begin{proposition}\label{prop:hard}
%[Corollary of the result in \cite{askalidis2013socially}]
Assuming UGC, it is NP-hard to approximate \kernel\ with interval orders within $\frac{3}{2}-\delta$ for any $\delta>0$. The statement holds even for the maximum stable marriage problem with interval orders.
%even if matroids are partition matroids with capacities $1$.
\end{proposition}

We then show that the $1.5$-approximability shown in Theorem~\ref{thm:main-approx} cannot be extended to the setting with general partial orders, even in the stable marriage case.
In Section \ref{sec:inapprox}, we investigate \maxsmp, the version of maximum stable marriage where preferences can be arbitrary partial orders. 
Note that for this problem and also for the more general \kernel, $2$-approximation is easily attained; as explained before, a stable matching (a matroid kernel) can be found by running the (matroidal version of) Gale--Shapley algorithm with total orders consistent with the original partial orders. Since stability implies that the output matching (common independent set) is maximal, its size is at least half of the optimum. 
We show that for general partial orders, we cannot beat this trivial approximation ratio even for the stable marriage case under UGC.
\begin{theorem}
\label{thm:maxsmp}
Assuming UGC, it is NP-hard to approximate \maxsmp\ within $2-\delta$ for any $\delta >0$.
\end{theorem}
%We prove that under the assumption that the Unique Games Conjecture is true, it is NP-hard to approximate the problem within $2-\delta$, for any $\delta>0$. 
The proof of Theorem~\ref{thm:maxsmp} uses a result of Bansal and Khot \cite{bansal2009optimal} about the hardness of finding an independent set of size $\epsilon n$ in a graph that has two disjoint independent sets of size $(\frac{1}{2}-\epsilon)n$ each. 
We also show in Section~\ref{sec:inapprox} that the integrality gap of the LP relaxation of a natural IP of \maxsmp\ has integrality gap at least $2$ for general partial orders while it is at most $1.5$ for interval orders as mentioned above. 
%This shows that the LP rounding approach is useless for general partial prders even when UGC is not assumed
Our results suggest that maybe the class of interval orders is  the right generalization for which a nontrivial approximation is still possible.

\subsection{Related work}\label{sec:related} % Yu: Please feel free to edit.
The \emph{stable marriage problem with ties and incomplete lists} ({\sc smti}) was first studied by Iwama et al.~\cite{IMMM99}, who showed the NP-hardness of \maxsmti. Since then, various algorithms have been proposed to improve the approximation ratio \cite{halldorsson2007improved,IMY07, IMY08, Kiraly11}, and the current best ratio is $1.5$ by a polynomial-time algorithm of McDermid \cite{Mcdermid09}, where the same ratio is attained by linear-time algorithms of Paluch \cite{paluch2011faster, Paluch14} and Kir\'aly \cite{kiraly2012linear, Kiraly13}.
%The current best lower bound is $\frac{33}{29}$ \cite{Yanagisawa07}. 
%Moreover, the integrality gap of a natural IP formulation for {\sc max-smti} is shown to be at least $\frac{3}{2}-o(1)$ in \cite{IMY14}, which rules out the possibility of further improvement by techniques such as rounding and primal-dual algorithms.
The $1.5$-approximability extends to the many-to-one matching setting \cite{Kiraly13}
and the student-project allocation problem with ties \cite{cooper20183}.
As for the inapproximability of \maxsmti, Halldórsson et al. \cite{halldorsson2002inapproximability} showed that it is NP-hard to approximate it within some constant factor. Later, inapproximability results have been improved, especially assuming stronger complexity-theoretic conjectures. Yanagisawa \cite{yanagisawa2007approximation} and Huang et al. \cite{huang2015tightsmti} showed that assuming the Unique Games Conjecture, there is no $(\frac{4}{3}-\varepsilon)$-approximation for any $\varepsilon >0$, if P$\ne$NP. By a recent work by Dudycz, Manurangsi and Marcinkowski \cite{smti1.5inapprox}, it follows that, assuming the Strong Unique Games Conjecture or the Small Set Expansion Hypothesis, there cannot even be a $(\frac{3}{2}-\varepsilon)$-approximation algorithm for \maxsmti, if P$\ne$NP. 
%If there can be free edges, that is, edges that can be included but cannot block (this can be modeled with interval order preferences), even $(\frac{3}{2}-\varepsilon )$ approximation is UGC-hard \cite{askalidis2013socially}.

The stable marriage problem also has generalizations in which constraints are imposed on the possible sets of edges. Bir{\'o} et al.~\cite{BFIM10} studied the {\em college admissions problem with common quotas}. It is a many-to-one matching problem between students and colleges, where not only individual colleges but also certain subsets of colleges have upper quotas. Huang \cite{CCH10} investigated the {\em classified stable matching problem}, a many-to-one matching model in which each individual college has upper and lower quotas for subsets of students. In the absence of lower quotas, it was shown that a stable matching exists in these models if the constraints have a laminar structure. 
%As these models consider strict preferences, the laminar cases (without lower quotas) can be reduced to the matroid framework of Fleiner \cite{fleiner2001matroid}. 
Yokoi~\cite{yokoi2021approximation} considered a many-to-many matching model with ties and laminar constraints and presented a $1.5$-approximation algorithm. Its approximation analysis depends on the base orderability of laminar matroids and cannot extend to the general matroid setting.

The stable matching problem has been studied with many different types of preference structures. 
%Irving, Manlove and Scott \cite{irving2003strong} and Condon et al. \cite{rastegari2014reasoning} studied the problem with arbitrary partial order preferences instead of weak orders. 
There are studies on super and strong stabilities with partial order preferences \cite{irving2003strong, rastegari2014reasoning}.
Condon et al. \cite{rastegari2014reasoning} studied super stability with partial orders, where a blocking edge is defined as an edge $e=(u,w)$ such that each of $u$ and $w$ either strictly prefers the other to its current assignment, or are indifferent between them (i.e. the two are incomparable in the partial order). 
They showed that 
%while super stable matchings may not always exists, 
we can find a super stable matching in polynomial-time if one exists. 
%Furthermore, they still admit a lattice structure, and there is an ``optimal'' super stable matching for both sides.
Irving, Manlove, and Scott \cite{irving2003strong} studied strong stability, which differs from super stability in the aspect that an edge $e=(u,w)$ only blocks a matching if at least one of $u,w$ strictly prefers the other to their current assignment. They showed that with arbitrary partial orders in the preferences, it is NP-hard to decide if a strongly stable matching exists. With some modifications to their reduction from 3-SAT, it can also be shown that this hardness holds even in the case of semiorders, which is a special case of interval orders. 

A closely related area is stable matchings with uncertain or changing preferences, where the preferences of the agents may be partially unknown, or may depend on some random factors and may change over time, which has also been the focus of interest lately \cite{aziz2020stable,chen2018stable,miyazaki2019jointly,bredereck2020adapting}. The usual purpose here is to find matchings that are stable with probability one, if there are any, or otherwise find matchings that are stable with maximum probability. In particular, as Aziz et al.\ \cite{aziz2020stable} mentioned, uncertain preferences are strongly connected to partial order preferences. For example, if we suppose that there is a set of possible preference lists for each agent, then deciding if there is a matching that is stable with any possible choices for these preference lists can be reduced to finding a super stable matching where the agents have partial orders. 

Another similar problem is robust stable matchings and locally nearly stable matchings \cite{chen2021matchings}. Intuitively speaking, a robust stable matching is one that is stable, and remains stable even if each agent is allowed to make some swaps (i.e. switch two adjacent entries) in their preference lists, while a locally nearly stable matching is a matching that can be made stable if each agent is allowed to make some swaps in their list, or equivalently, there are no blocking edges, where both agents improve by a lot. Finding robust stable matchings reduces to finding super stable matchings, while finding a maximum size locally nearly stable matching reduces to finding a maximum size (weakly) stable matching with partial order preferences (in this case the arising orders are semiorders). 

%Cardinal preferences in the context of stable matchings have been studied for several reasons. Pini et al. \cite{pini2013stability} analyzed manipulations consisting of falsely reporting preference values. Among other stability notions, they introduced \emph{$\alpha$-stability}, which is equivalent to our definition of $\Delta$-min-stability.  Anshelevich et al. \cite{anshelevich2013anarchy} considered approximate stability from the point of view of social welfare. They defined various utility models, and $\alpha$-stability with respect to these models. They gave price-of-anarchy bounds that depend on the value of $\alpha$. Chen et al. \cite{chen2021matchings} defined \emph{local $d$-near-stability} and \emph{global $d$-near-stability} based on swaps in the preference orders. The former can be seen as a special case of our $\Delta$-min-stability notion with $\Delta$ being integer and the preferences being ordinal. They proved hardness of approximation results regarding the minimum value $d$, such that there exists a locally/globally $d$-nearly-stable matching that is complete or has welfare higher than some number $\eta$. 

\section{Existence of perfect matching of exchange edges in matroids} \label{sec:improving}

In this section, we present our key tool, a result on exchange properties of matroids. 

A {\em matroid} is a pair $(S, \cI)$ of a finite set $S$ and a nonempty family $\cI\subseteq 2^S$ satisfying the following two axioms: (i) $A\subseteq B\in \cI$ implies $A\in \cI$, and (ii) for any $A, B\in \cI$ with $|A|<|B|$, there is an element $x\in B\setminus A$ with $A+x\in \cI$. A set in $\cI$ is called an {\em independent set}, and an inclusion-wise maximal one is called a {\em base}. By axiom (ii), all bases have the same size, which is called the {\em rank} of the matroid.
%Let $M=(S,\cI)$ be a matroid, where $S$ is the ground set and $\cI$ is the family of independent sets. Recall that a base is a maximum size independent set, while 
A {\em circuit} is an inclusion-wise minimal dependent set. The \emph{fundamental circuit} of an element $x\in S\setminus B$ for a base $B$, denoted by $C_B(x)$, is the unique circuit in $B+x$. By a slight abuse of notation, we will also use $C_I(x)$ for an independent set $I$ and an element $x\in S\setminus I$ to denote the unique circuit in $I+x$ if it exists. Any pair of circuits satisfies the following property.

\begin{proposition}[Strong circuit axiom]
If $C,C'$ are circuits, $x\in C\setminus C'$, and $y\in C \cap C'$, then there is a circuit $C'' \subseteq C \cup C'$ such that $x \in C''$ and $y \notin C''$. \qed 
\end{proposition}

If we have a total order $\succ$ given on $S$, then the triple $M=(S,\cI,\succ)$ is called a \emph{(totally) ordered matroid}. A nice property of totally ordered matroids is that for any weight vector $w\in \Rset^S$ that satisfies $w_x>w_y \Leftrightarrow x \succ y$, the unique maximum weight base is the same. We call this base $A$ the \emph{optimal base} of $(S,\cI,\succ)$; it is characterized by the property that the worst element of $C_A(x)$ is $x$ for any $x \in S \setminus A$. Equivalently, a base $A$ is optimal with respect to $\succ$ if and only if any pair of elements $a\in A$ and $b\in S\setminus A$ with $A-a+b\in \cI$ satisfies $a\succ b$.

Here we  prove the theorem that will be our main tool in proving the approximation bound for our algorithm. To our knowledge, this result on exchanges has not been previously observed in the literature. See Remark~\ref{rem:exchange} for a comparison between a previously known property. We use the notation $[r]=\{1,2,\dots,r\}$.

\begin{theorem} \label{thm:matching}
Let $M=(S,\cI,\succ)$ be a totally ordered matroid of rank $r$. Let $A$ be the optimal base and $B$ be a base disjoint from $A$. Then, there is a perfect matching $a_ib_i$ $(i\in [r])$ between $A$ and $B$ such that $a_i\succ b_i$ and $B+a_i-b_i \in \cI$ for every $i \in [r]$. 
\end{theorem}
\begin{proof}
Define $E_A=\{\,ab\in A \times B:\, A-a+b \in \cI\,\}$ and $E_B=\{\,ab\in A \times B:\, B+a-b \in \cI\,\}$.
\begin{claim}\label{cl:1}
Let $C$ be any circuit with $C\subseteq A\cup B$. For any element $a\in C\cap A$, there exists $b\in C\cap B$ with $ab\in E_A$. For any element $b\in C\cap B$, there exists $a\in C\cap A$ with $ab\in E_B$.
\end{claim}
\begin{proof}
We only show the first claim because the second one is shown symmetrically. Take any $a\in C\cap A$. Suppose conversely that $ab\not\in E_A$ for any $b\in C\cap B$. This means that, we have $a\not\in C_A(b)$ for any $b\in C\cap B$. Set $C'=C$ and repeatedly update it as follows while $C'\cap B\neq \emptyset$: (1) Take any $b\in C'\cap B$, (2) apply the strong circuit axiom to $C', C_A(b), a\in C'\setminus C_A(b)$, and $b\in C'\cap C_A(b)$ to obtain a circuit $C''$ satisfying $C''\subseteq C'\cup C_A(b)$, $a\in C''$, and $b\not\in C''$, (3) update $C'$ by $C''$. Then, $C'$ always satisfies $a\in C'\subseteq A\cup B$ and the size of $|C'\cap B|$ decreases monotonically. We finally obtain a circuit $C'$ with $C'\subseteq A$, a contradiction.   
\end{proof}

Define $E=\{\,ab\in E_B:\, a\succ b\,\}$. 
Then, showing the existence of a perfect matching in $G=(A, B;E)$ completes the proof of the theorem.
To this end, we show the following claim.
\begin{claim}\label{cl:2}
For any circuit $C$ with $C\subseteq A\cup B$, there exists $ab\in E$ with $\{a,b\}\subseteq  C$.
\end{claim}
\begin{proof}
By Claim~\ref{cl:1}, we see that there is a cycle $Q\subseteq E_A\cup E_B$ (possibly of length two) that is contained in $C$ and uses edges in $E_A$ and $E_B$ alternately. (Start at any element in $C$ and apply Claim~\ref{cl:1} repeatedly.) As $A$ is optimal, any $a'b'\in Q\cap E_A$ satisfies $a'\succ b'$. Then, there must exist a pair $ab\in Q\cap E_B$ with $a\succ b$ because otherwise the relation $\succ$ on the elements in $Q$ would become cyclic and could not form a total order. This $ab$ belongs to $E$. 
%when we see $Q$ as a directed cycle on which $Q\cap E_A$ edges are passed through from $B$ to $A$ and $Q\cap E_B$ edges are passed through $A$ to $B$, the 
\end{proof}
We now show the existence of a perfect matching in $G=(A, B;E)$.
Suppose, to the contrary, there is no perfect matching in $G$. 
By Hall's theorem, then there exists a set $X \subseteq A$ such that $|X|>|\Gamma_G(X)|$, where $\Gamma_G(X)=\{b\in B: \exists a \in X, ab \in E\}$.
The size of the set $X\cup (B\setminus \Gamma_G(X))$ is larger than that of the base $B$, and hence there exists a circuit $C\subseteq X\cup (B\setminus \Gamma_G(X))$. By Claim~\ref{cl:2}, there exists $ab\in E$ included in $C$, which satisfies $a\in X$ and $b\in B\setminus \Gamma_G(X)$,  contradicting the definition of $\Gamma_G(X)$.
\end{proof}

\begin{remark}\label{rem:exchange}
It is a well-known fact that, for any two bases $A, B$ of a matroid $(S,\cI)$ of rank~$r$, there exists a perfect matching $a_ib_i$ $(i\in [r])$ between $A$ and $B$ such that $B+a_i-b_i\in\cI$ for every $i\in [r]$ (see Brualdi~\cite{brualdi1969comments} and also \cite[Corollary 39.12a]{schrijver2003combinatorial} and \cite[Theorem 5.3.4]{frank2011connections}). We claim that this property does not immediately imply our Theorem~\ref{thm:matching}. 

In our theorem, the base $A$ is assumed to be optimal. This implies the condition $a\succ b$ for all pairs $ab\in A\times B$ with $A-a+b\in \cI$ but not for those with $B+a-b\in \cI$. To see this, consider the graphic matroid of $K_4$ with ground set $S=\{e_1, e_2, e_3, e_4, e_5, e_6\}$ shown in Figure~\ref{ex:graphic}. Suppose that the total order $\succ$ is defined as $e_1\succ e_2\succ e_3\succ e_4\succ e_5\succ e_6$. Then, $A\coloneqq\{e_1, e_2, e_4\}$ is the optimal base with respect to $\succ$ and its complement $B\coloneqq S\setminus A$ is also a base. Here, we have $B+e_4-e_3\in \cI$ for $e_4\in A$ and $e_3\in B$ while $e_4\not\succ e_3$. 

Therefore, the existence of a perfect matching of exchangeability edges combined with the optimality of $A$ does not simply imply Theorem~\ref{thm:matching}.
%In particular, $B$ is not the worst base.
\begin{figure}[bhtp]\label{ex:graphic}
\centering
\begin{tikzpicture}[xscale=0.9, yscale=.9]
      \tikzset{enclosed/.style={draw, circle, inner sep=0pt, minimum size=.15cm, fill=black}}
      \node[enclosed] (T) at (3,6) {};
      \node[enclosed] (L) at (1,2.8) {};
      \node[enclosed] (R) at (5,2.8) {};
      \node[enclosed] (C) at (3,4) {};
      \draw[line width=0.5mm] (T) -- (L) node[midway, left] (edge2) {$e_1$};
      \draw (T) -- (C) node[pos=0.7, left] (edge2) {$~~e_3\!$};
      \draw (T) -- (R) node[midway, right] (edge2) {$e_5$};
      \draw[line width=0.5mm] (C) -- (L) node[pos=0.3, below] (edge2) {$e_2$};
      \draw[line width=0.5mm] (C) -- (R) node[pos=0.4, below] (edge2) {$e_4$};
      \draw (L) -- (R) node[midway, below] (edge2) {$e_6$};
\end{tikzpicture}
\caption{The graphic matroid defined by $K_4$. When $e_1\succ e_2 \succ e_3\succ e_4\succ e_5\succ e_6$,  the thick edges form the optimal base with respect to $\succ$.}
\label{fig1}
\end{figure}

\end{remark}

\section{Matroid kernel algorithm with interval orders}\label{sec:kernel}

In this section, we show Theorem~\ref{thm:main-approx}, which states that \kernel\ is approximable with a factor $1.5$ if the partial orders are interval orders, i.e., $(2+2)$-free posets. 

%Our algorithm is a generalization of the algorithm presented in \cite{yokoi2021approximation}, which considered a generalization of \maxsmti\ that included laminar constraints. However, that algorithm can deal with only weak orders and the proof of the approximation ratio there relied crucially on the property that the matroids induced by the laminar constraints are base orderable. In contrast, our algorithms works for interval orders, a much broader class of preferences, and our proof works for arbitrary matroids.
Like the previous algorithms by Yokoi \cite{yokoi2021approximation} and by the present authors \cite{csaji2022approximation} (conference version of this paper), our algorithm is described as an application of the Gale--Shapley algorithm to a carefully constructed modified instance. In this modified instance, each element is replicated into three parallel elements, and special total orders are defined on the extended ground set. The origin of the idea comes from Király's $1.5$-approximation algorithm \cite{Kiraly11} for \maxsmti, which is a variant of the Gale--Shapley algorithm in which each man can propose to each woman at most three times and there are special rules for men's proposal order and women's acceptance/rejection. 

Our construction of the modified instance is symmetric for the two sides, and the $1.5$-approximation ratio depends only on the stability in the modified instance (not on the behaviour of the GS algorithm). This allows for a simple analysis while considering broad class preferences and constraints.

\subsection{Description of the algorithm}\label{sec:construction}
Let $M_1=(S,\cI_1, \succ_1)$ and $M_2=(S,\cI_2, \succ_2)$ be partially ordered matroids on the same ground set $S$ and suppose that $\succ_1$ and $\succ_2$ are interval orders. 
Our algorithm consists of three steps. 
\begin{enumerate}
\item The algorithm first creates a new instance by replacing each element of $S$ by three parallel elements, and by defining total orders on the extended ground set as explained below.
%based on the original partial orders. 
\item For the obtained totally ordered matroids $M^*_1$, $M^*_2$, an $(M^*_1,M^*_2)$-kernel $A^*$ is computed. 
\item The algorithm returns a set $A$ that is the projection of $A^*$ to the original ground set $S$.
\end{enumerate}
The first step can be done in $\mathcal{O}(|S|^2)$ time as explained later, and the second step can be done in $\mathcal{O}(|S|^2)$ time using Fleiner's algorithm, which we will briefly describe later.

Here, we explain the precise construction of the totally ordered matroids $M^*_1$ and $M^*_2$. Let the extended ground set be 
$S^*\coloneqq\cup_{u \in S} \{x_u,y_u,z_u\}$. We define $M^*_i=(S^*,\cI^*_i,\succ^*_i)$ for $i=1,2$ as follows. The elements $x_u,y_u,z_u$ are parallel in each $M^*_i$, that is,
\[\cI^*_i=\{\,I^* \subseteq S^*: \pi(I^*)\in \cI_i,\ \ |I^*\cap\{x_u,y_u,z_u\}|\leq 1\ (\forall u \in S)\,\},\]
where $\pi(I^*)=\{\,u\in S: I^*\cap \{x_u,y_u,z_u\} \neq \emptyset\,\}$. It is easy to see that $(S^*, \cI^*_i)$ is a matroid for $i=1,2$. To define the total order $\succ^*_1$ on the extended ground set $S^*$, we first define 
a binary relation $\mathcal{R}^*_1\subseteq S^*\times S^*$ using $\succ_1$ as follows:
\begin{align*}
\mathcal{R}_1^*=\quad&~\{\, (x_u, z_w) : u, w\in S\,\}\\
\cup &~\{\, (y_u, x_w), (y_u,y_w), (y_u, z_w) : \,u, w\in S, \,\,u\succ_1 w\,\}\\
\cup &~\{\, (x_w, y_u) : \,u,w\in S, \,u\not\succ_1 w\,\}.    
\end{align*}
To define $\succ^*_2$, we define $\mathcal{R}^*_2$ from $\succ_2$ as follows, which is described in the same manner as $\mathcal{R}^*_1$ but the roles of $x_u$ and $z_u$ are interchanged.
\begin{align*}
\mathcal{R}_2^*=\quad&~\{\, (z_u, x_w) : u, w\in S\,\}\\
\cup &~\{\, (y_u, z_w), (y_u,y_w), (y_u, x_w) : \,u, w\in S, \,\,u\succ_2 w\,\}\\
\cup &~\{\, (z_w, y_u) : \,u,w\in S, \,u\not\succ_2 w\,\}.    
\end{align*}

For each $i=1,2$, we let $\succ^*_i$ be an arbitrary total order on $S^*$ such that $s\succ^*_i t$ holds for any $(s,t)\in \mathcal{R}^*_i$. The existence of such a total order is guaranteed by the following lemma. 
\begin{lemma}\label{lem:acyclic}
Let $D_i$ be a digraph with vertex set $S^*$ such that $D_i$ has an arc $(s,t)\in S^*\times S^*$ if and only if $(s, t)\in \mathcal{R}^*_i$. Then $D_i$ is acyclic.
\end{lemma}
The proof of this lemma, which uses the fact that $\succ_i$ is an interval order, is postponed to the end of this section. We now complete the construction of the new instance.

Since the digraph $D_i$ is acyclic and has $\mathcal{O}(|S|^2)$ arcs, we can define a total order $\succ_i^*$ on $S^*$ consistent to $D_i$ by topological sort in $\mathcal{O}(|S|^2)$ time. This $\succ_i^*$ indeed satisfies the required condition, i.e., $s\succ^*_i t$ holds for any $(s,t)\in \mathcal{R}^*_i$.
This completes the construction of the totally ordered matroids $M^*_1$ and $M^*_2$.

For completeness, here we describe how Fleiner's algorithm works for our new instance $M_1^*$ and $M^*_2$.
The algorithm first sets $R$ to be an empty set and repeats the following three steps: set $X$ to be the optimal base of $M^*_1$ with the ground set restricted to $S^*\setminus R$,
%the maximum $p^*_1$-weight independent subset of $S\setminus R$ (with respect to $M_1^*)$, 
set $Y$ to be the optimal base of $M^*_2$ with the ground set restricted to $X$,
%the maximum $p^*_2$-weight independent subset of $X$ (with respect to $M_2^*)$, 
and update $R$ with $R\cup (X\setminus Y)$. The repetition stops if $Y=X$, and $Y$ is returned, which is an $(M^*_1,M^*_2)$-kernel (see \cite{aziz2022matching} for this version of the description). 
%The number of repetition is at most $|S^*|=\mathcal{O}(|S|)$ and each round runs in $\mathcal{O}(|S|^*)=\mathcal{O}(|S|)$ time.

Here, we show that the output of the algorithm is indeed a matroid kernel (i.e., a stable common independent set) in the original instance. The approximation ratio is shown in the next section.
We use $C^1_I(u)$ and $C^2_I(u)$ to denote fundamental circuits in $M_1$ and $M_2$, respectively.
\begin{lemma}\label{lem:stability}
The output of our algorithm is 
%a stable common independent set of $M_1$ and $M_2$.
an $(M_1, M_2)$-kernel.
\end{lemma}
\begin{proof}
Let $A=\pi(A^*)$ be the output of the algorithm, where $A^*$ is the $(M^*_1,M^*_2)$-kernel given by Fleiner's algorithm. Since $A^*\in \cI^*_1\cap \cI^*_2$, it is clear from the definitions of $\cI^*_i$ that $A \in \cI_1\cap \cI_2$. Suppose for contradiction that there exists $u\in S\setminus A$ that blocks $A$; we claim that $y_u$ blocks $A^*$. As $u$ blocks $A$, for each $i=1,2$, we have $A+u \in \cI_i$ or $u \succ_i  v$ holds for some $v \in C^i_A(u)$. In the former case, $A^*+y_u \in \cI^*_i$ immediately follows. In the latter case, $u \succ_i  v$ implies that $v^*\coloneqq\{x_v, y_v, z_v\}\cap A^*$ satisfies $(y_u, v^*)\in \mathcal{R}^*_i$, and therefore $y_u\succ^*_i v^*$ by the construction of $\succ^*_i$, while $v^*$ belongs to the fundamental circuit of $y_u$ for $A^*$ in $M^*_i$.
Thus, $y_u$ blocks $A^*$.
\end{proof}

We now provide the postponed proof.

\begin{proof}[Proof of Lemma~\ref{lem:acyclic}]
Since the constructions of $\mathcal{R}^*_1$ and $\mathcal{R}^*_2$ are symmetric, it is sufficient to show that the digraph $D_1$ defined from $\mathcal{R}^*_1$ is acyclic.

Suppose, to the contrary, that $D_1$ has directed cycles. Let $C$ be a cycle of minimum length.
By the definition of $\mathcal{R}^*_1$, we can observe the following properties. We say that an element $s\in S^*$ is an $x$-element (resp., $y$-element, $z$-element) if $s=x_u$ (resp., $s=y_u$, $s=z_u$) for some $u\in S$.
\begin{itemize}
\item Every $z$-element has no leaving arcs. So, $C$ consists of only $x$-elements and $y$-elements.

\item No arc connects two $x$-elements. So in $C$ there are no consecutive $x$-elements.

\item There are no consecutive $y$-elements in $C$. Indeed, if $y_{u_1}, y_{u_2}, x_{u_3}$ (resp., $y_{u_1}, y_{u_2}, y_{u_3}$) appear in $C$ in this order consecutively, then $D_1$ has arcs $(y_{u_1}, y_{u_2})$ and $(y_{u_2}, y_{u_3})$ (resp., $(y_{u_2}, x_{u_3})$), and hence $u_1\succ_1 u_2$ and $u_2\succ_1 u_3$. Then, the transitivity of $\succ_1$ implies $u_1\succ_1 u_3$, and hence $D_1$ has an arc $(y_{u_1}, x_{u_3})$ (resp, $(y_{u_1}, y_{u_3})$), contradicting the minimality of $C$.

\item Therefore, $C$ has an even length and $x$-elements and $y$-elements appear alternately. 

\item Suppose that four distinct elements $y_{u_1}, x_{u_2}, y_{u_3}, x_{u_4}$ appear in $C$ in this order consecutively. 
Then we have $u_1\neq u_3$, $u_2\neq u_4$, $u_1\succ u_2$, and $u_3\succ u_4$. 

\begin{itemize}
\item If $u_1\neq u_4$ and $u_2\neq u_3$, then all $u_i$ are distinct. Since $\succ_1$ is an interval order (i.e., $(2+2)$-free), then the relations $u_1\succ u_2$ and $u_3\succ u_4$ imply that $u_1\succ u_4$ or $u_3\succ u_2$. Hence there exists an arc $(y_{u_1}, x_{u_4})$ or $(y_{u_3}, x_{u_2})$. 

\item If $u_1=u_4$ (resp., $u_2=u_3$), then $u_3\succ u_4=u_1\succ u_2$ (resp., $u_1\succ u_2=u_3\succ u_4$), and hence there exists an arc $(y_{u_3}, x_{u_2})$ (resp. ($y_{u_1}, x_{u_4}$)).
\end{itemize}
The existence of an arc $(y_{u_1}, x_{u_4})$ or $(y_{u_3}, x_{u_2})$ contradicts the minimality of $C$.
\end{itemize}
Then $C$ must be of length two, but it is impossible by the definition of $\mathcal{R}^*_1$. 
\end{proof}

\subsection{Proof of $1.5$-approximation}\label{sec:apxproof}

We now show the approximation ratio of our algorithm.
As in the previous section, we denote by $C^1_I(u)$ and $C^2_I(u)$ fundamental circuits in $M_1$ and $M_2$, respectively.
\begin{theorem}
\label{thm:approx-for-any-matroid}
	The approximation ratio of the above algorithm is at most $1.5$.
\end{theorem}

\begin{proof}
 Let $A=\pi(A^*)$ be the output of the algorithm, where $A^*$ is an $(M^*_1,M^*_2)$-kernel, and  let $B$ be a largest $(M_1,M_2)$-kernel.  
 Suppose for contradiction that $|B|>1.5 |A|$. Let $B_i$ be a subset of $B\setminus A$ such that $A \cup B_i \in \mathcal{I}_i$ and $|A \cup B_i|=|B|$ for each $i\in \{1,2\}$. (The existence of such $B_i$ follows from axiom (ii) of matroids.) 
 The sets $B_1$ and $B_2$ are disjoint because $A^*$ is an inclusion-wise maximal common independent set of $M^*_1$ and $M^*_2$. In the following, we say that an element $u\in A$ is \emph{of type} $x$ (resp., $y,z$) if $\{x_u,y_u,z_u\}\cap A^*=x_u$ (resp., $y_u,z_u$).
    \begin{lemma}\label{lem:matchings}
    Let $i\in \{1,2\}$. There is a matching $N_i$ of size $|B_{3-i}|$ between $A\setminus B$ and $B_{3-i}$  such that the following hold for every $uv \in N_i$, where $u\in A$ and $v\in B$: 
	\begin{enumerate}
		\item $u$ is of type $x$ or $y$ if $i=1$, and of type $y$ or $z$ if $i=2$
		\item %$v\not\succ_i u$, and in particular 
  $u\succ_i v$ if $u$ is of type $y$
		\item either $v \in C^i_B(u)$ or $B+u \in \mathcal{I}_i$. 
	\end{enumerate}
    \end{lemma}
\begin{proof}
Let $M'=(S',\cI')$ be the matroid obtained from $M_i$ by deleting $S\setminus (A \cup B)$, contracting $(A\cap B) \cup B_i$, and truncating to the size of $A\setminus B$. That is, $S'=(A\setminus B)\cup (B\setminus (A\cup B_i))$ and $\cI'=\{\, I\subseteq S': I\subseteq A\cup B,~ I\cup (A\cap B)\cup B_i \in \cI_i,~ |I|\leq |A\setminus B|\,\}$. In $M'$, the sets $A'\coloneqq A\setminus B$ and $B'\coloneqq B \setminus (A \cup B_i)$ are bases that are complements of each other. 

We define a total order $\succ'$ on $S'$ as follows. The elements of $B \setminus (A \cup B_i \cup B_{3-i})$ are worst (in arbitrary order). On the remaining elements, i.e., on the elements of $(A\setminus B)\cup B_{3-i}$, we define the preferences based on the total order $\succ^*_i$ on $S^*$. To do this, we assign an element $u^*\in S^*$ to each $u \in (A\setminus B)\cup B_{3-i}$ as follows. Let $u^*=\{x_u,y_u,z_u\} \cap A^*$ if $u \in A\setminus B$, let $u^*=x_u$ if $i=1$ and $u\in B_2$, and let $u^*=z_u$ if $i=2$ and $u\in B_1$. We then let $u \succ' v$ if and only if $u^* \succ^*_i v^*$. In the totally ordered matroid $M'=(S',\cI',\succ'$), $A'$ is an optimal base. Indeed, $v$ is the worst element of $C_{A'}(v)$ for every $v\in B'$. It is clear for the elements in $B'\setminus B_{3-i}$ by the definition of $\succ'$. As for each $v\in B_{3-i}$, since $A^*+v^*\in \mathcal{I}^*_{3-i}$ holds and $A^*$ is an $(M^*_1,M^*_2)$-kernel, $v^*$ must be the worst element of its fundamental circuit for $A^*$. By Theorem \ref{thm:matching}, there is a perfect matching $N'$ between $A'$ and $B'$ such that $u \succ' v$ and $v \in C'_{B'}(u)$ for every $uv \in N'$, where $u \in A'$ and $v\in B'$. 

Let $N_i$ be the subset of $N'$ induced by $A \cup B_{3-i}$. Then $|N_i|=|B_{3-i}|$, and every $uv \in N_i$ satisfies $u \succ' v$, which implies $u^* \succ^* v^*$, which in turn implies $(v^*, u^*)\not\in \mathcal{R}_i^*$. In case $i=1$, as $v\in B_{3-i}=B_2$, we have $v^*=x_v$, and hence $(x_v,u^*)\not\in \mathcal{R}^*_1$. By the definition of $\mathcal{R}^*_1$, this implies the first two properties of the lemma. We can similarly obtain these two in the case $i=2$.

Now we show that for every $uv \in N_i$, either $v \in C^i_B(u)$ or $B+u \in \mathcal{I}_i$. Since $v \in C'_{B'}(u)$, $v$ is in the fundamental circuit of $u$ for $B'$ in the matroid obtained by truncating $M_i$ to the size of $A\setminus B$. This means that it is either in the fundamental circuit also in $M_i$, or $B+u$ is independent in $M_i$, as required.
	\end{proof}

We are now ready to prove the theorem by obtaining a contradiction. 
Let $N_1$ and $N_2$ be matchings described in Lemma~\ref{lem:matchings}. 
Since $|B|>1.5|A|$ implies $|N_i|=|B_{3-i}|>|A\setminus B|/2$ for $i\in \{1,2\}$, there is an element $u \in A\setminus B$ that is covered by both $N_1$ and $N_2$. Let $uv_1 \in N_1$ and $uv_2 \in N_2$. Since the first two properties of Lemma~\ref{lem:matchings} hold for $i\in \{1,2\}$, the element $u$ must be of type $y$, and $u \succ_i v_i$ for $i\in \{1,2\}$. But this means that $u$ blocks $B$ because of the third property of Lemma~\ref{lem:matchings}, a contradiction. 
\end{proof}

Thus, we have completed the proof of the $1.5$-approximability of \kernel\ with interval orders.

We conclude this section by observing that Proposition~\ref{prop:hard}, i.e., the UGC-hardness of $(1.5-\delta)$-approximation, follows from the result of Askalidis et al. \cite{askalidis2013socially} on the {\em maximum stable marriage problem with incomplete lists under social stability} (\textsc{max-smiss}). An input of the problem is a graph $G=(U,W;E)$, a total order $\succ_v$ on $\delta_G(v)$ for each $v\in U\cup W$, and a set $A\subseteq E$ whose elements are called {\em acquainted pairs}. For a matching $M\subseteq E$, an edge $e\in E$ {\em socially blocks} $M$ if it blocks $M$ in the classical sense and belongs to $A$. A matching is {\em socially stable} if there is no social blocking pair. The task of \textsc{max-smiss} is to find a maximum socially stable matching. Askalidis et al. \cite{askalidis2013socially} proved that, assuming UGC, it is NP-hard to approximate \textsc{max-smiss} within $\frac{3}{2}-\delta$ for any $\delta>0$.

We can reduce \textsc{max-smiss} to the stable marriage problem with interval orders as follows. 
Given an instance of \textsc{max-smiss}, from the total order $\succ_v$ of each $v\in U\cup W$, define a partial order $\succ'_v$ by $e\succ'_v f\Leftrightarrow e\succ_v f, e\in A$. 
Then $\succ'_v$ is an interval order. Indeed, we can easily observe that, if $e\succ'_v f$ and $g\succ'_v h$ hold for distinct $e,f,g,h$, then 
%we have $e,g\in A$, $e\succ_v f$, $g\succ_v h$, and then either $e\succ_v h$ or $g\succ_v f$ as $\succ_v$ is a total order, and hence 
$e\succ'_v h$ or $g\succ'_v f$.
It is also not hard to see that a socially stable matching in the original instance is a stable matching with respect to the interval orders $\{\succ'_v\}_{v\in U\cup W}$. Thus, the maximum stable marriage with interval orders and its generalization to \kernel\ are inapproximable within $\frac{3}{2}-\delta$ under UGC.

\section{Inapproximability with General Partial Orders}\label{sec:inapprox}
In this section, we consider the case where preferences may be arbitrary partial orders. 
In particular, we investigate \maxsmp, the maximum stable marriage problem with partial orders.
Interestingly, in this case the structure of the problem changes significantly compared to the case with weak orders (i.e., \maxsmti), as shown in Section~\ref{sec:structure}. In fact, beating the trivial $2$-approximation becomes UGC-hard, 
%as we show in Theorem \ref{thm:maxsmp} 
as shown in Section~\ref{sec:UGChardness}. 

\subsection{Structural observations}\label{sec:structure}
Here we provide some structural observations on \maxsmp.
The first one shows that the main tool which is used to show the $1.5$-approximability in previous approaches for \maxsmti\ 
cannot be used for partial order preferences.
The previous $1.5$-approximation algorithms for \maxsmti, such as one in \cite{Kiraly13}, are designed so that, for the output stable matching $M$ and any stable matching $N$, there is no maximal alternating path in $M\cup N$ consisting of one $M$-edge and two $N$-edges. In the case of arbitrary partial orders, it may happen that for any stable matching $M$, there exists another stable matching $N$, such that a maximal alternating path with one $M$-edge and two $N$-edges exists, as the next example shows.

\begin{figure}[h]
    \centering
    \includegraphics[width=0.5\textwidth]{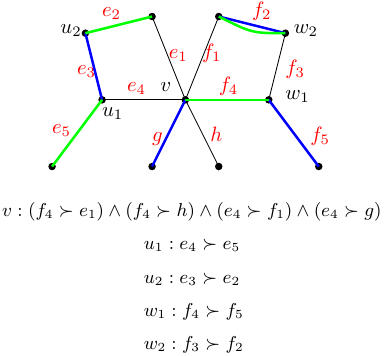}
    \caption{Illustration of Example \ref{ex:weird}. The green stable matching is $M_2$ and the blue one is $M_3$.}
    \label{fig:weird}
\end{figure}

\begin{example}
\label{ex:weird}
    Consider the instance in Figure~\ref{fig:weird}. We claim that, in this example, for any stable matching $M$, there is another stable matching $N$, such that in $M\cup N$ there is a maximal alternating path of length 3, with two $N$-edges and one $M$-edge. % such that the middle $M$ edge does not block $N$. 
    It is easy to verify that the following matchings are stable: $M_1=\{ e_2,e_4,f_2,f_5\}, M_2=\{ f_2,f_4,e_2,e_5\}, M_3=\{ g,e_3,f_2,f_5\},M_4=\{ h,f_3,e_2,e_5\}$.

    Consider vertex $v$, which must be matched in any stable matching $M$. If $f_1\in M$ or $g\in M$ (resp., $e_1\in M$ or $h\in M$), then $e_3\in M$ (resp., $f_3\in M$) too since otherwise $e_4$ (resp., $f_4$) blocks $M$, so $M\cup M_2$ (resp., $M\cup M_1$) contains a desired alternating path $e_2 e_3 e_5$ (resp., $f_2 f_3 f_5$). If $f_4\in M$ (resp., $e_4\in M$), then $M\cup M_3$ (resp., $M\cup M_4$) contains a desired alternating path $gf_4f_5$ (resp., $he_4e_5$). 
    %Finally, if $g\in M$ (resp., $h\in M$), then $e_3\in M$ (resp., $f_3\in M$) too since otherwise $e_4$ (resp., $f_4$) blocks $M$, so $M\cup M_2$ (resp., $M\cup M_3$) contains such an alternating path $e_2 e_3 e_5$. 
\end{example}

The second observation is about the integrality gap. Observe that
\maxsmp\ on $G=(U, W;E)$ can be modeled with the following integer programming problem (IP) with variables $\{p_e\}_{e\in E}$.
Since we allow parallel edges, to distinguish edges connecting the same vertices, say $u$ and $w$, we use a notation $(u,w)_i$ where distinct edges have distinct subscripts $i$: 
%We denote $\delta_G(v)$ simply as $\delta(v)$ if there is no confusion:

\begin{align}
\text{max} \quad\sum_{e\in E}~p_e \hspace{60mm}&&&\nonumber\\
\text{s.t.}\hspace{60mm} \sum_{e\in \delta_G(v)}p_e&\le 1 & &  \forall v\in U\cup W \label{cond1}\\
%\sum_{e':e\not{\succ_u}e'\in \delta_G(u) \text{or} e\not{\succ_w}e'\in \delta_G(w)}p_{e'}\ge 1 & & \forall e=(u,w)_i\in E \\
\sum_{e'\in \delta_G(u):e\not{\succ_u}e'}\hspace{-2mm}p_{e'}+\sum_{e'\in \delta_G(w):e\not{\succ_w}e'}\hspace{-2mm}p_{e'}-\sum_{\substack{\scriptstyle e'\in \delta_G(u)\cap \delta_G(w): \\\scriptstyle e\not{\succ_u}e', e\not{\succ_w} e'}} \hspace{-2mm}p_{e'}&\ge 1 & & \forall e=(u,w)_i\in E \label{cond2}\\
 p_e\in& \{ 0,1\} & & \forall e\in E \nonumber
\end{align}

In the special case of \maxsmti, that is, the case with weak orders, it is known that the integrality gap of the linear programming (LP) relaxation of this IP is at least $1.5-o(1)$ \cite{IMY14}. This fact is sometimes considered to indicate a potential barrier to improving the approximation ratio $1.5$.

We show that this integrality gap can be $2$ for the general \maxsmp\ while it is at most $1.5$ even for interval orders. First, we give an example showing the former claim.
%, whereas it is known to be at most $1.5$ in the case of \maxsmti. 

\begin{example}
    Consider an instance of \maxsmp\ with vertices $u,u',w,w'$ and edges $e_i=(u,w')_i$ for $i\in [2]$, $f_i=(u',w)_i$ for $i\in [2]$, and $g_j=(u,w)_j$ for $j\in [4]$. The partial orders of $u$ and $w$ are given by $(g_1\succ_u e_1)\wedge (g_2\succ_u e_1)\wedge (g_3\succ_u e_2)\wedge (g_4\succ_u e_2)$ and $(g_1\succ_w f_1)\wedge (g_2\succ_w f_2)\wedge (g_3\succ_w f_1)\wedge (g_4\succ_w f_2)$. The orders of $u'$ and $w'$ are defined arbitrarily.

    In this instance, $p_{e_1}=p_{e_2}=p_{f_1}=p_{f_2}=0.5, \ p_{g_j}=0~(j\in [4])$ is a solution to the LP relaxation of the above IP, but any stable matching must have size one (it must contain one of the $g_j$ edges). Thus, the integrality gap is $2$.
\end{example}

We next show that, in contrast to this example, the integrality gap is at most $1.5$ for interval orders. The proof depends on the correctness of our $1.5$-approximation algorithm shown in Section~\ref{sec:kernel}.

\begin{theorem}\label{thm:smti-int-gap}
For a \maxsmp\ instance consisting of $G=(U,W;E)$ and $\{\succ_v\}_{v\in U\cup W}$, if $\succ_v$ are interval orders, then the output of our algorithm $1.5$-approximates the LP optimum of the relaxation of the above IP. In particular,
the integrality gap of the IP is at most $1.5$.
\end{theorem}
\begin{proof}
Since \maxsmp\ with interval orders is a special case of \kernel\ with interval orders (recall the remark just after Theorem~\ref{thm:main-approx}), we can apply our algorithm in Section~\ref{sec:kernel} to the given instance.
Let $M^*\subseteq \cup_{e\in E}\{x_e, y_e, z_e\}$ be the stable matching in the corresponding modified instance and $M$ be its projection to $E$, i.e., the algorithm's output.
Let $\{p^*_e\}_{e\in E}\in [0,1]^{E}$ be an optimal solution to the LP. 
We show $\sum_{e\in E}p^*_E\leq \frac{3}{2}|M|$, which completes the proof.

For an edge $e=(u,w)_i\in M$, let $F_e\subseteq E$ be the set of edges $f$ such that one of $f$'s endpoints is $u$ or $w$ and the other is uncovered by $M$. Then, $F_e\subseteq \delta_G(u)\cup \delta_G(w)$ and $F_e\cap (\delta_G(u)\cap \delta_G(w))=\emptyset$.
We now show $\sum_{f\in F_e}p^*_f\leq 1$. We can assume $F_e\cap \delta_G(u)\neq \emptyset$ and $F_e\cap \delta_G(w)\neq \emptyset$ since otherwise the claim is trivial. 
For any $f\in F_e\cap \delta_G(u)$, its copy $x_f$ in the modified instance is incident to $u$ and the other endpoint is uncovered by $M^*$, and hence the stability of $M^*$ implies $e^* \succ^*_1 x_f
 $, where $e^*$ is the $e$'s copy in $M^*$ and $\succ^*_1$ is the $U$-side total order in the modified instance. Similarly, for any $g\in F_e\cap \delta_G(w)$, we obtain $e^*\succ^*_2 z_g$, where $\succ^*_2$ is the $W$-side total order. By the construction of $\succ^*_1, \succ^*_2$, these imply $e^*=y_e$ and that we have $e\succ_u f$ for every $f\in F_e\cap \delta_G(u)$ and $e\succ_w g$ for every $g\in F_e\cap \delta_G(w)$.
Thus, in the LP, every element in $F_e$ appears either in inequality \eqref{cond1} for $u$ or that for $w$, and does not appear in the inequality \eqref{cond2} for $e$. By adding the inequalities \eqref{cond1} for $u$ and $w$ and subtracting \eqref{cond2} for $e$ from it, we obtain $\sum_{f\in F_e}p^*_f\leq 1$ as required.

Let $\partial M$ denote the set of vertices in $U\cup W$ covered by $M$. Observe that $F\coloneqq \cup_{e\in M} F_e$ coincides with the set of edges connecting $\partial M$ and $(U\cup W)\setminus \partial M$, and that $\{F_e\}_{e\in M}$ forms a partition of $F$. Let $\Delta\coloneqq \sum_{f\in F}p^*_f=\sum_{e\in M}\sum_{f\in F_e}p^*_f\leq |M|$, where the last inequality follows from the fact $\sum_{f\in F_e}p^*_f\leq 1$ shown above.

Note that, as $M$ is stable, it is maximal, and hence every edge in $E$ has at least one endpoint in $\partial M$. Therefore, every $e\in E\setminus F$ has both endpoints in $\partial M$. 
Then, by summing the inequalities \eqref{cond1} for all vertices in $\partial M$, we obtain \,$2\sum_{e\in E\setminus F}p^*_e+\sum_{f\in F}p^*_f\leq |\partial M|=2|M|$, which implies $\sum_{e\in E\setminus F}p^*_e\leq |M|-\frac{1}{2}\Delta$. 
Combined with $\Delta=\sum_{f\in F}p^*_f\leq |M|$, this implies $\sum_{e\in E}p^*_e\leq |M|-\frac{1}{2}\Delta+\Delta=|M|+\frac{1}{2}\Delta\leq \frac{3}{2}|M|$.
\end{proof}

\subsection{UGC-hardness of $(2-\delta)$-approximation}\label{sec:UGChardness}
In this section, we provide our hardness reduction to prove Theorem~\ref{thm:maxsmp}, i.e., we show that it is UGC-hard to beat the trivial $2$-approximation for \maxsmp\ with general partial orders. We use the following theorem of Bansal and Khot about \is.

\begin{theorem}[Bansal and Khot \protect\cite{bansal2009optimal}]
\label{thm:UGC}
    Assuming UGC, for any $\varepsilon >0$ it is NP-hard, given an $n$-vertex graph that has two disjoint independent sets of size $(\frac{1}{2}-\varepsilon )n$ each, to find an independent set of size $\varepsilon n$.
\end{theorem}

%Now we are ready to prove our main theorem of the section.

\begin{comment}
\begin{theorem}
\label{thm:maxsmp}
    Assuming UGC, it is NP-hard to approximate \maxsmp\ within $2-\delta$ for any $\delta >0$.
\end{theorem}
\end{comment}

\begin{proof}[Proof of Theorem~\ref{thm:maxsmp}]
We show that a $(2-\delta)$-approximation algorithm for \maxsmp\ implies that we can find an independent set of size at least $\varepsilon n$ in an instance of Theorem~\ref{thm:UGC}, if $\varepsilon$ is small enough, which is a contradiction. 

Let $G=(V,E)$ be an instance of \is, and let $V=(v_1,v_2,\dots,v_n)$. %such that there are two disjoint independent sets $V_1,V_2\subset$ with $|V_1|=|V_2|=(\frac{1}{2}-\varepsilon )n$.
We create an instance $(G'=(U,W;E'), \{\succ_v\}_{v\in U\cup W})$ of \maxsmp\ as follows. 
For clarity, we refer to the elements of $V$ and $U\cup W$ as ``vertices'' and ``agents,'' respectively.
For each vertex $v_i\in V$, we create four agents $u_i, y_i\in U$ and $w_i, x_i\in W$ and create six edges $g_{i(\ell)}=(u_i,w_i)_{\ell}$ for $\ell\in[2]$, $h_{i(\ell)}=(u_i,x_i)_{\ell}$ for $\ell\in [2]$, and $h'_{i(\ell)}=(y_i,w_i)_{\ell}$ for $\ell\in [2]$. Then, for each $(v_i,v_j)\in E$ we create four edges $e_{ij}=(u_i,w_j)_1, f_{ij}=(u_i,w_j)_2, e_{ji}=(u_j,w_i)_1,f_{ji}=(u_j,w_i)_2$. This completes the construction of the bipartite graph $G'$. (See Figure~\ref{fig:constr} for an example.)

We describe the relations in the partial orders of the agents in Table~\ref{tab:my_label}.

\begin{table}[h]
\label{tab:POprefs}
    \centering
    \begin{tabular}{cc}
       $u_i:$ & $(g_{i(1)}\succ [F(u_i)]\succ h_{i(1)})$ $\wedge$ $(g_{i(2)}\succ [E(u_i)]\succ h_{i(2)})$ \\
       $w_i:$ & $(g_{i(2)}\succ [F(w_i)]\succ h'_{i(1)})$ $\wedge$ $(g_{i(1)}\succ [E(w_i)]\succ h'_{i(2)})$
    \end{tabular}
    \smallskip
    \caption{Partial orders used in the proof of Theorem~\ref{thm:maxsmp}.}
    \label{tab:my_label}
\end{table}

Here, $[X(v_i)]$ for $v_i\in \{ u_i,w_i\}$ and $X\in \{ E,F\}$ denotes a strict ranking over the adjacent edges to $v_i$ of type $X$ according to the indices of the other endpoint (the smaller index is ranked higher). 

This concludes the construction part. The construction is illustrated in Figure \ref{fig:constr}

\begin{figure}[h]
    \centering
    \includegraphics[scale=1]{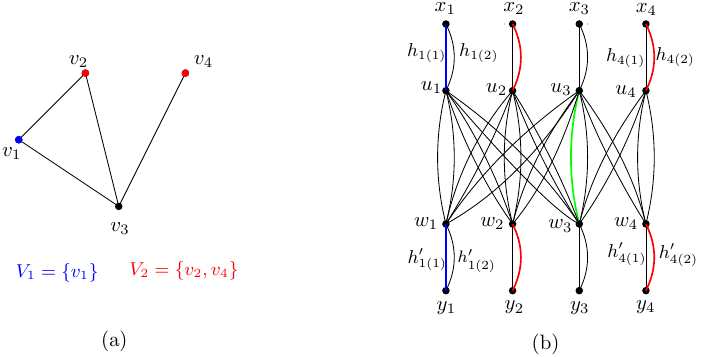}
    \caption{The reduction for Theorem \ref{thm:maxsmp}. (a) An instance of \is\ with two disjoint independent sets $V_1$ (blue) and $V_2$ (red). (b) The corresponding instance of \maxsmp. Parallel edges $h_{i(\ell)}$ and $h'_{i(\ell)}$ with $\ell=1$ are drawn on the left and those with $\ell=2$ are drawn on the right. Blue and red edges represent the edges corresponding to the two independent sets $V_1,V_2$, respectively, which are chosen in the constructed matching.}
    \label{fig:constr}
\end{figure}

\begin{claim}\label{clm:reduction1}
    If there are disjoint independent subsets $V_1,V_2\subset V$ with $|V_1|+|V_2|=k$ in $G$, then there is a stable matching of size $n+k$ in $G'$.
\end{claim}
\begin{proof}
    Let $V_1,V_2$ be two disjoint independent sets in $G$. We create a matching $M$ in $G'$ as follows. For each $v_i\in V_1$, we add the edges $h_{i(1)}=(u_i,x_i)_1, h'_{i(1)}=(y_i,w_i)_1$. For each $v_i\in V_2$, we add the edges $h_{i(2)}=(u_i,x_i)_2, h'_{i(2)}=(y_i,w_i)_2$. Finally, for $v_i\in V\setminus (V_1\cup V_2)$, we add an edge $g_{i(1)}=(u_i,w_i)_1$. Clearly, $M$ has size $2k+(n-k)=n+k$.

    We claim that $M$ is stable. First, observe that each $u_i,w_i$ agent is matched in $M$. No $g_{i(\ell)}$ type edge can block $M$ because such a blocking edge would imply that there is an index $i\in [n]$ such that $h_{i(\ell)}, h'_{i(3-\ell)}\in M$ for some $\ell \in [2]$, which contradicts the construction of $M$. No $e_{ij}$ or $f_{ij}$ type edge can block $M$ either, because $e_{ij}$ (resp., $f_{ij}$) could block $M$ only if $h_{i(\ell)}, h'_{j(\ell)}\in M$ for some $\ell \in [2]$, however, such indices $i,j$ must satisfy $v_i,v_j\in V_{\ell}$ by the construction, and as $V_1$ and $V_2$ were independent sets, such an edge $e_{ij}$ (resp., $f_{ij}$) could not exist in the first place. Any $h_{i(\ell)}$ (resp., $h'_{i(\ell)}$) type edge cannot block $M$ either, as its endpoint $u_i$ (resp., $w_i$) is covered by a better or indifferent edge in $M$. 
\end{proof}

\begin{claim}
\label{claim:2}
    If there is a stable matching of size $n+k$ in $G'$, then there are disjoint independent sets $V_1,V_2\subset V$ in $G$ such that $|V_1|+|V_2|=k$.
\end{claim}
\begin{proof}
Let $M$ be a stable matching of size $n+k$. 

We claim that for any $i\in [n]$ and $\ell\in [2]$, if $h_{i(\ell)}\in M$, then $h'_{i(\ell)}\in M$. 
Suppose for contradiction that $h_{i_0(1)}\in M$ but $h'_{i_0(1)}\notin M$ for some $i_0\in [n]$. Since $g_{i_0(1)}=(u_{i_0},w_{i_0})_1$ does not block $M$ while $u_{i_0}$ prefers it to $h_{i_0(1)}\in M$, agent $w_{i_0}$ must be covered by some edge that is not worse than (i.e., better than or incomparable with) $g_{i_0(1)}$. As $h'_{i_0(1)}\notin M$, then $f_{i_1i_0}\in M$ must hold for some $i_1\in [n]$. Now, we know that $G'$ has an edge $f_{i_0i_1}$. Since $f_{i_0i_1}=(u_{i_0}, w_{i_1})_2$ and $g_{i_1(1)}=(u_{i_1}, w_{i_1})_1$ do not block $M$ while $u_{i_0}$ prefers $f_{i_0i_1}$ to $h_{i_0(1)}\in M$ and $u_{i_1}$ prefers $g_{i_1(1)}$ to $f_{i_1i_0}\in M$, agent $w_{i_1}$ must be covered by some edge that is not worse than $f_{i_0i_1}$ and not worse than $g_{i_1(1)}$. Thus, we obtain that $f_{i_2i_1}\in M$ for some $i_2<i_0$. Now, since $f_{i_1i_2}=(u_{i_1}, w_{i_2})_2$ and $g_{i_2(1)}=(u_{i_2}, w_{i_2})_1$ do not block $M$ while $u_{i_1}$ prefers $f_{i_1i_2}$ to $f_{i_1i_0}\in M$ (as $i_2<i_0$) and $u_{i_2}$ prefers $g_{i_2(1)}$ to $f_{i_2i_1}\in M$, we get that $w_{i_2}$ is covered by an edge no worse than $f_{i_1i_2}$ and $g_{i_2(1)}$, and hence $f_{i_3i_2}\in M$ for some $i_3<i_1$. By the same argument, from the fact that $f_{i_2i_3}=(u_{i_2},w_{i_3})_2$ and $g_{i_3(1)}=(u_{i_3},w_{i_3})_1$ do not block $M$, we get that $f_{i_4i_3}\in M$ for some $i_4<i_2<i_0$. By iterating this argument, we get that there must be infinitely many $f_{ij}$ type edges in $M$, which is a contradiction. 

By using similar arguments, we get that, for any $i\in [n]$ and $\ell\in[2]$, we have $h_{i(\ell)}\in M$ if and only if $h'_{i(\ell)}\in M$. 
Recall that the size of $M$ is $n+k$ and observe that all $u_i, w_i$ agents are covered in $M$ (since otherwise $h_{i(\ell)}$ or $h'_{i(\ell)}$ type edges block $M$).
Let $V_1=\{\, v_i\in V : h_{i(1)}\in M\,\}$ and $V_2=\{\, v_i\in V :  h_{i(2)}\in M\,\}$. Then, $V_1\cap V_2=\emptyset$ and the fact $h_{i(\ell)}\in M\Leftrightarrow h'_{i(\ell)}\in M$ implies that $|V_1|+|V_2|=k$. Suppose that there is an edge $(v_i,v_j)\in E$ with $v_i, v_j\in V_1$. Then, the edge $f_{ij}$ blocks $M$, a contradiction. Similarly, if there is an edge $(v_i,v_j)\in E$ with $v_i, v_j\in V_2$, then $e_{ij}$ blocks $M$, a contradiction. Thus, $V_1$ and $V_2$ are disjoint independent sets with $|V_1|+|V_2|=k$ as required.
\end{proof}

Suppose for contradiction that there is a polynomial-time $(2-\delta)$-approximation algorithm for \maxsmp\ for some fixed $\delta >0$.
For any $\varepsilon >0$, let $G$ be an $n$-vertex graph such that there are disjoint independent sets with size $(\frac{1}{2}-\varepsilon )n$ each.
We have shown in Claim~\ref{clm:reduction1} that there is a stable matching of size $2n-2\varepsilon n$ in the corresponding \maxsmp\ instance $G'$. 
Using a $(2-\delta)$-approximation algorithm for \maxsmp, we can find a stable matching in $G'$ of size at least $\frac{1}{2-\delta}(2-2\varepsilon )n$, where we have $\frac{1}{2-\delta}(2-2\varepsilon )n>(1+2\varepsilon)n$ if $\varepsilon$ is small enough (for example $\varepsilon <\frac{\delta}{6}$). 
Then, by Claim~\ref{claim:2},  we can obtain two disjoint independent sets $V_1$ and $V_2$ in $G$ with $|V_1|+|V_2|> 2\varepsilon n$. Then, we get that we can find an independent set of size at least $\varepsilon n$, which contradicts Theorem \ref{thm:UGC}.
\end{proof}

\section*{Acknowledgement}
Some of our results were obtained at the Emléktábla Workshop in Gárdony, July 2022. We would like to thank Tamás Fleiner, Zsuzsanna Jankó, and Ildikó Schlotter for the fruitful discussions. We thank the anonymous reviewers of the previous versions for their helpful feedback. The work was supported by the Lend\"ulet Programme of the Hungarian Academy of Sciences -- grant number LP2021-1/2021, by the Hungarian National Research, Development and Innovation Office -- NKFIH, financed under the ELTE TKP2021-NKTA-62 funding scheme and grant K143858. The first author was supported by the Ministry of Culture and Innovation of Hungary from the National Research, Development and Innovation fund, financed under the KDP-2023 funding scheme (grant number C2258525).
The last author was  supported by JST PRESTO Grant Number JPMJPR212B and JST ERATO Grant Number JPMJER2301, and the joint project of Kyoto University and Toyota Motor Corporation,titled “Advanced Mathematical Science for Mobility Society”. 

%\section*{Declarations}
%\section*{Compliance with ethical standards}
%\paragraph{Conflict of interest}
%The authors declare that they have no conflict of interest.
\clearpage
\bibliographystyle{plain}
\bibliography{cit}
\end{document}